\documentclass[journal]{IEEEtran}  
\usepackage{amsmath,amsfonts}
\usepackage{amsmath}
\usepackage{amsthm}
\usepackage{amssymb} 
\usepackage{array}
\usepackage[caption=false,font=normalsize,labelfont=sf,textfont=sf]{subfig}
\usepackage{textcomp}
\usepackage{booktabs}       
\usepackage{multirow}   
\usepackage{multicol}   
\usepackage{stfloats}
\usepackage{url}
\usepackage{verbatim}                   
\usepackage{xcolor}
\usepackage{graphicx}
\usepackage[colorlinks, citecolor=black,linkcolor=black ]{hyperref}      
\usepackage{cite}
\usepackage[ruled,linesnumbered]{algorithm2e}
\usepackage{booktabs}       
\usepackage{verbatim}      
\makeatletter

\newcommand{\Rmnum}[1]{\expandafter@slowromancap\romannumeral #1@}
\makeatother

\begin{document}
	\color{black} 
	\title{Interference Management in MIMO-ISAC Systems: A Transceiver Design Approach }
	
	\author{\IEEEauthorblockN{Yangyang Niu, \emph{Student Member, IEEE}, Zhiqing Wei, \emph{Member, IEEE}, Dingyou Ma, \emph{Member, IEEE}, Xiaoyu Yang,  Huici Wu, \emph{Member, IEEE}, Zhiyong Feng, \emph{Senior Member, IEEE} and Jianhua Yuan, \emph{Member, IEEE}}
		
		\thanks{
			
			Y. Niu, Z. Wei, D. Ma, X. Yang, H. Wu, Z. Feng  and Jianhua Yuan are with Beijing University of Posts and Telecommunications, Beijing 100876, China (e-mail: niuyy@bupt.edu.cn, weizhiqing@bupt.edu.cn,  dingyouma@bupt.edu.cn, xiaoyu.yang@bupt.edu.cn, dailywu@bupt.edu.cn, fengzy@bupt.edu.cn, jianhuayuan@bupt.edu.cn).   }}
	
	
	
	\maketitle
	\thispagestyle{empty}
	\pagestyle{empty}
	\begin{abstract}\label{abstract}
		The integrated sensing and communication (ISAC) system under multi-input multi-output (MIMO) architecture achieves dual functionalities of sensing and communication on the same platform by utilizing spatial gain, which provides a feasible paradigm facing spectrum congestion. However, the dual functionalities of sensing and communication operating simultaneously in the same platform bring severe interference in the ISAC systems. Facing this challenge, we propose a joint optimization framework for transmit beamforming and receive filter design for ISAC systems with MIMO architecture. We aim to maximize the signal-to-clutter-plus-noise ratio (SCNR) at the receiver while considering various constraints such as waveform similarity, power budget,  and communication performance requirements to ensure the integration of the dual functionalities. In particular, the overall transmit beamforming is refined into sensing beamforming and communication beamforming, and a quadratic transformation (QT) is introduced to relax and convert the complex non-convex optimization objective. An efficient algorithm based on covariance matrix tapers (CMT) is proposed to restructure the clutter covariance matrix considering the mismatched steering vector, thereby improving the robustness of the ISAC transceiver design. Numerical simulations are provided to demonstrate the effectiveness of the proposed algorithm.
	\end{abstract}
	

	\section{Introduction}\label{introduction}
	\IEEEPARstart{T}{he} vision of intelligent connectivity for everything in sixth generation (6G) mobile communication system is continuously driving the emergence of numerous emerging applications, such as autonomous driving, holographic communication, intelligent cities, and immersive extended reality \cite{b1}. The continuous development trends of wireless services and the demand for high spectrum utilization efficiency have promoted a new design principle, which integrates wireless communication and sensing functions into one system, so-called integrated sensing and communication (ISAC) \cite{b2}. This design pattern has sparked widespread research interest in the joint design of communication and sensing \cite{b3, b4, wei2023integrated, feng2020joint, wang2021symbiotic}.  
	
	An ISAC system simultaneously performs sensing and communication functions, achieving mutual benefits through effective sharing of the same spectrum resources, integrated transceiver platform design, and joint signal processing framework \cite{cui2021integrating, 10273396, liu2022survey}. However, the communication and sensing functionalities are coupled together. The enhancements in communication performance influence the sensing capability and vice versa. The diverse interference, including multi-user interference (MUI) for communication, clutter generated by multiple targets in the environment, and sensing interference (SI) between sensing and communication, bring additional complexity to the ISAC signal processing, which brings challenges to achieving high-reliability information transmission and high-precision sensing in ISAC systems \cite{ma2020joint}. With multiple-input-multiple-output (MIMO) architecture widely employed for enabling ISAC systems, combining multiple beams utilizing spatial degrees of freedom enables communication with multiple users while simultaneously sensing multiple targets. Hence, by utilizing beamforming design, feasible solutions can be designed to avoid interference. Therefore, this paper focuses on the robust design of MIMO-ISAC systems for interference management. 
	
	The robust design for the MIMO-ISAC systems mainly solves two subproblems. The first subproblem is interference avoidance at the transmitter, and the second is interference suppression at the receiver. Initially, the two subproblems are studied separately. 
	For the first subproblem, studies can be classified into three categories based on the considered interference types: MUI \cite{a13,9424454,liu2020joint,zhuo2022multibeam }, clutter \cite{li2017joint, chen2022generalized, aldayel2016successive }, and SI \cite{liu2020joint, li2017joint, zhuo2022multibeam}.
	Liu  \textit{et al.} proposed a design of transmit waveform by minimizing a joint least squares function of waveform similarity error (WSE) and total MUI. The weighting factors are adjusted to balance the performance between dual functions \cite{a13}. 
	In the scenario with a single target and multiple users, Chen \textit{et al.}  proposed a Pareto optimization framework to design the joint beamforming based on the peak side-lobe level metric for radar sensing and the signal-to-interference-plus-noise ratio (SINR) metric for communication \cite{9424454}. However, this approach is computationally complex and does not consider the interference between sensing and communication. 
	In \cite{aldayel2016successive}, the sensing design incorporates similarity and constant-modulus constraints. However, the communication interference is assumed to be fixed. To overcome this challenge, Liu \textit{et al.} proposed a streamlined approach that incorporates zero-forcing MUI and SI to mitigate the impacts of sensing and communication \cite{liu2020joint}. They creatively derived the separated sensing and communication beamforming matrices at the transmitter, thereby significantly enhancing the flexibility of the transmit beamforming design. 
	
	Regarding the second subproblem, the design principle is to achieve a distortion-free output of signals in the desired direction. The research can generally be classified into single-receive filter design \cite{o2017relaxed, blunt2010embedding} and multiple-receive filter design \cite{sahin2017filter}.	
	O’Rourke \textit{et al.} jointly optimized the signal pulse and received filter to effectively suppress clutter under the assumption of minimum variance distortionless response (MVDR) with the prior knowledge of clutter channel state information (CSI)  \cite{o2017relaxed}. Blunt \textit{et al.} proposed a filter design method based on the least squares method \cite{blunt2010embedding}. However, these filters exhibit a high sidelobe level (SLL) and experience a degradation in SINR. To overcome this problem, a cascaded filtering method is proposed in \cite{sahin2017filter}. This approach involves the design of two consecutive filters. The primary objective of the first filter is to enhance the similarity between different filter outputs, while the second filter aims to reduce the SLL. Although this approach effectively mitigates range sidelobe modulation and reduces the SLL in the filter outputs, it results in a degraded SINR. 
	
	Further, the superiority of the joint design of transmit beamforming and receive filter is demonstrated under several constraints \cite{cheng2019joint, tang2016joint,  o2020quadratic,  liu2022achieving,liu2022joint}. For instance, Chen \textit{et al.} described a joint design of transmit and receive beamforming under energy constraint and similarity constraint to maximize the  sensing SINR, so as to enhance the probability of extended target detection. An iterative scheme is devised via the semidefinite relaxation (SDR) technique and MVDR to solve the resultant problem \cite{cheng2019joint }. 
	Liu \textit{et al.} presented an optimal transceiver beamforming algorithm, which leverages the sequential convex approximation (SCA) approach to obtain the Pareto boundary of the performance region. Within the boundary, it is possible to achieve the simultaneous service and detection of multiple communication users and radar sensing targets \cite{ liu2022achieving}.
	Liu \textit{et al.} jointly designed the transmit waveform and the receive filter to maximize the output SINR while satisfying the communication quality-of-service (QoS) requirements and the constant modulus power constraint \cite{liu2022joint}.
	
	Notice that all the aforementioned studies have generally relied on a common assumption of perfect prior information of CSI and the targets' directions at the transmitter. Due to the factors such as feedback delay and limited feedback bandwidth, the transmitter can only acquire incomplete information in practical applications. Hence, it becomes essential to pursue robust transceiver design. Two types of definitions for robustness are commonly found in the existing literature, specifically, “worst-case robustness” 
	\cite {wang2009worst,wang2013robust} and “statistical robustness” \cite{zhang2009robust, zhu2015robust, tsinos2021joint, shen2013joint }. Specifically, Zhu  \textit{et al.} proposed a robust joint design  to maximize the worst-case SINR of targets with an unknown statistical distribution \cite{ zhu2015robust }. Although the scenario with the maximum angular information loss of targets is considered, the objective is still to maximize the sensing output SINR. In other words, the impact of the angular information on clutter suppression is not addressed.  Moreover, the optimization model does not take the communication process into consideration. Tsinos  \textit{et al.} introduced a novel joint transceiver design method for multi-user MIMO (MU-MIMO) system, which incorporates matching filter combination vectors in beamforming design using channel time correlation \cite{tsinos2021joint}. This approach improves the traditional signal-leakage-plus-noise ratio (SLNR) beamforming performance and focuses on the robustness to mismatched CSI. From a statistical standpoint, achieving robustness can be pursued by optimizing the averaged performance metric when the statistical distribution of uncertainties is known.  
	Though the given robust schemes can achieve appreciable performance gains, there is still a significant gap in the design of robust joint transceiver architectures with strong interference resistance for MIMO-ISAC systems to achieve desired output SINR and transmit beampattern property. 
	
	\begin{table*}[t]
		\centering
		\caption {comparison of this work with existing literature on transceiver design for MIMO-ISAC systems}
		\label{table}
		\begin{tabular}{|c|c|c|c|c|c|c|c|c|c|c|c|c|c|c|c|}  
			\hline  
			& [12 ]               & [14]              & [16]   & [17]       
			& [18] & [19]  & [21]&[22] & [25]&   [26] & [28]& [30]& [31] & [32] & Our work \\  
			\hline  
			MUI & \checkmark & \checkmark & \checkmark  & \checkmark
			& $\times$  & $\times$  & $\times$ & $\times$  & \checkmark&   \checkmark & $\times$ & \checkmark & \checkmark & \checkmark &\checkmark \\  
			\hline
			
			SI & $\times$  & \checkmark & \checkmark & $\times$
			& $\times$  & $\times$  &$\times$ & $\times$ & $ \times $& $ \times $ &$\times$ & $\times$ & $\times$ & $\times$ &\checkmark \\  
			\hline 
			
			Clutter & $\times$ &\checkmark & \checkmark &\checkmark
			& \checkmark  & \checkmark  & \checkmark &\checkmark& \checkmark&  \checkmark &$\times$& \checkmark & \checkmark & $\times$ &\checkmark\\  
			\hline  
			Mismatched steering vector & $\times$ & $\times$ & $\times$  & $\times$
			& $\times$  & $\times$  &$\times$&$\times$& $\times$&   $\times$ & \checkmark & \checkmark & $\times$ & \checkmark &\checkmark \\  
			\hline
			
			Transmitter design & \checkmark  & \checkmark & \checkmark & \checkmark
			& \checkmark  & $\times$ &   $\times$  &\checkmark& \checkmark&   \checkmark & \checkmark & \checkmark & \checkmark & \checkmark &\checkmark \\  
			\hline
			
			Receiver design & $\times$  & $\times$ & $\times$ &\checkmark
			& $\times$   & \checkmark & \checkmark &\checkmark & \checkmark&   \checkmark &$\times$& \checkmark & \checkmark &\checkmark &\checkmark\\  
			\hline
			
			Transceiver design & $\times$  & $\times$ & $\times$ & \checkmark
			& $\times$  & $\times$  & $\times$ & \checkmark & \checkmark &   \checkmark  & $\times$& \checkmark & \checkmark & \checkmark &\checkmark \\  
			\hline
			
			MIMO-ISAC systems & \checkmark  & \checkmark & $\times$ & \checkmark
			& $\times$  & $\times$ & \checkmark &$\times$& \checkmark& \checkmark &$\times$&$\times$ & \checkmark & $\times$ &\checkmark\\   
			\hline
		\end{tabular}
	\end{table*}
	
	This paper presents a novel framework to enhance target detectability in signal-dependent clutter and multiple communication users through the robust joint transceiver design in MIMO-ISAC systems. In Table \uppercase\expandafter{\romannumeral1}, we provide a brief comparison of this study with other existing works in the literature to highlight its discreteness. The main contributions of this paper are as follows.
	\begin{itemize}
	\item[$\bullet$] \textbf{Interference Signal Modeling: } 
	\textcolor{black}{ In downlink ISAC systems, the communicating users are interfered by sensing signals in addition to the traditional multi-user interferences. Signal-dependent clutter caused by environmental targets arriving at the receiving array can suppress sensing accuracy. This paper reconstructs the clutter covariance matrix in the presence of mismatched steering vectors.}
	
	\end{itemize} 	
	\begin{itemize}
		\item[$\bullet$]\textbf{Transmit Beamforming Design:} Quadratic transform (QT) is used for the transmit beamforming subproblem due to the fractional objective functions. The problem can be transformed into an equivalent convex problem via SDR, and the tightness of the relaxation is presented. The Cholesky decomposition is used to derive the communication and sensing beamforming matrices.
	
	\end{itemize}
	\begin{itemize}
		\item[$\bullet$] \textbf{Receive Filters Design:} The covariance matrix taper (CMT) technique is used to modify the clutter covariance matrix. To derive an improved weight vector, we substitute the reconstructed clutter covariance matrix by Malloux-Zatman (MZ) taper into the MVDR filter. The CMT approach enables adjustable null width in clutter directions, which is practically significant compared to the ideal scenario of perfect CSI knowledge.
	\end{itemize}	
	\begin{itemize}
		\item[$\bullet$] \textbf{Performance Verification:} To comprehensively examine the influence of various parameters on the performance of sensing and communication, we investigate the compensation effect of multi-antenna gain through extensive observations. Under the framework of joint transceiver design, as the number of antennas increases, the null depths of clutter become deeper, allowing for greater sensing gains while satisfying communication QoS. 
	 
	\end{itemize}
	
	The subsequent sections of this paper are structured as follows. Section \uppercase\expandafter{\romannumeral2} introduces the system model and performance metrics for MIMO-ISAC  systems. Section \uppercase\expandafter{\romannumeral3} provides a detailed description of the bivariate optimization problem model under multiple constraints. The corresponding solutions, which utilize an alternating optimization approach involving beamforming design at the transmitter and filter design at the receiver, are provided in Section \uppercase\expandafter{\romannumeral4}. Simulation results are provided in Section \uppercase\expandafter{\romannumeral5}, followed by concluding remarks in Section \uppercase\expandafter{\romannumeral6}.

Notations: $\mathbf{A}$, $\mathbf{a}$, $a$ denote a matrix, vector, and variable, respectively. $\mathbf{I}$ represents the unit matrix, and the subscript denotes the dimension. $(\cdot)^{T}$, $(\cdot)^{H}$, and $(\cdot)^{-1}$ denote the transpose, Hermitian transpose, and inverse of a matrix, respectively. $|\cdot|$ denotes the element-wise absolute value.  $\mathbb{C}^{m \times n}$ is the set of complex-valued matrices. $\mathbb{E}(\cdot)$ denotes the statistical expectation. 
\textcolor{black}{$\mathcal{S}_{m}^{+}$ represents a set of positive semidefinite matrices of dimension $m$.
	$\|\cdot\|_F$ denotes the Frobenius norm of a matrix.}

$\mathbf{A} \circ \mathbf{B} $ is the Hadamard product of matrices  $\mathbf{A}$ and $\mathbf{B}$. 

Moreover, the main notations used throughout the paper are listed in Table \ref{Notation}.
	
	\begin{table}[t]
		\centering
		\caption {Notations Used Through the Paper}
		\label{Notation}
		\begin{tabular}{cl}
			\toprule   
			\textbf{Notation}  & \textbf{Definition}  \\
			\midrule       
			$N_t(Nr)$  &   Number of transmit (receive) antennas \\
			$K$    & Number of communication users \\	
			$P$   &  Number of clutter blocks \\		 
			$\mathbf{X}$   &   Transmit ISAC signal\\ 	
			$\mathbf{X}_0$ &   Reference waveform\\ 		
			$\mathbf{W}_s (\mathbf{W}_c)$  & Sensing (communication) beamforming matrix \\
			$\mathbf{W}$    &  Overal transmit beamforming matrix \\  
			$\boldsymbol{\omega} $ &    Receive filter   \\
			$\mathbf{h}_k $   & Downlink channel vector between transmitter and $k$-th user \\   
			$ \mathbf{A}(\theta)$ &   Steering vector at direction $\theta$\\
			$\mathbf{R}$     &  Transmit ISAC signal covariance matrix  \\  
			$\mathbf{R}_c$   &	Clutter covariance matrix \\
			$\mathbf{\tilde{e}}$     &	Phase disturbance vector \\
			${\mathbf{T}} $   &Taper matrix \\
			$\gamma_k$ &	 SINR at $k$-th communication user\\
			$\rho$   &   Signal-to-clutter-plus-noise ratio \\
			$\Delta $  &     Parameter determines the width of clutter nulls  \\
			$ \alpha$  &     Parameter determines the  waveform similarity  \\
			$ \Gamma$  &      Communication SINR threshold\\
			\bottomrule    
		\end{tabular}
	\end{table}
	
	\section{System Model and Performance Metrics}
	\subsection{System Model}\label{system model}
	
	\begin{figure}
	\centerline{\includegraphics[scale=0.6]{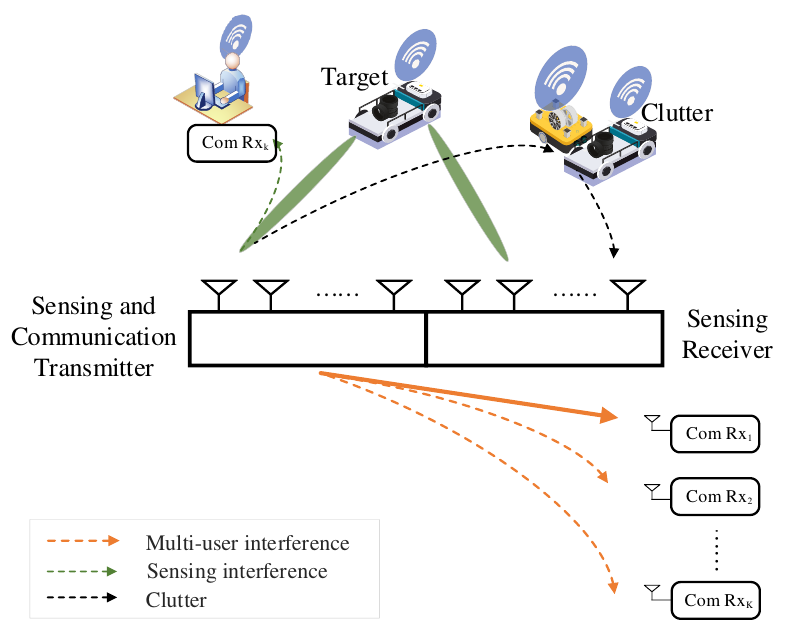}}
	\caption{The communication receiver is interfered with by sensing waveform and unintended communication symbols, while the sensing receiver is interfered with by clutter from environmental targets.}
	\label{Fig.1}
\end{figure}
The considered narrowband ISAC system is shown in Fig. 1, where an ISAC base station (BS) is equipped with $N_t$ transmit antennas and $N_r$ receive antennas arranged in a uniform linear array with half-wavelength inter-antenna spacing. For convenience, let $N_t=N_r= M $. 
\textcolor{black}{ The BS transmits the ISAC signal $\mathbf{X}$, simultaneously implementing the target sensing and $K$ users communication capabilities.  }
Each user is equipped with a single antenna that consecutively receives signals from the transmitter, where $K\leq M$. 
The multi-antenna platform with multi-beam control capability enhances the spatial freedom to support target sensing and multi-user communication. 
In particular, the transmit ISAC signal \textcolor{black}{$\mathbf{X} \in \mathbb{C}^{M\times N}$ } can be denoted by 
\begin{equation}
	\mathbf{X}= \mathbf{W}_{c} \mathbf{S}_{c}+\mathbf{W}_{s} \mathbf{S}_{s}=\left[\mathbf{W}_{c}, \mathbf{W}_{s}\right]\left[\begin{array}{l}
		\mathbf{S}_{c} \\
		\mathbf{S}_{s}
	\end{array}\right] =\mathbf{W S},
\end{equation}
\textcolor{black}{ where each row vector $\mathbf{x}[m] \in \mathbb{C}^{1 \times N }$ indicates a set of $N$ sampling results of the signal encoding sequence emitted by the $m$-th antenna, and each column vector $ \mathbf{x}[n] \in \mathbb{C}^{M\times 1}$ represents the signals transmitted by $M$ antenna elements at time index $n$.} 
$\mathbf{S}_c \in\mathbb{C}^{K \times N}$ denotes $ K$ parallel communication symbols for the desired signal and $\mathbf{S}_s \in\mathbb{C}^{M \times N}$ represents $M$ orthogonal sensing waveforms. 
$\mathbf{W}\in \mathbb{C}^{M \times (M+K)} $, $\mathbf{W}_{c}\in \mathbb{C}^{  M \times K} $, $\mathbf{W}_{s} \in\mathbb{C}^{M \times M}$ represent the overall beamforming matrix, the communication beamforming matrix, and the sensing beamforming matrix, respectively. 

\textcolor{black}{
	For a given time index $n$, the transmitted signal can be formulated as 
	\begin{equation}
		\mathbf{x}[n]= \mathbf{W}_{c} \mathbf{s}_{c}[n]+\mathbf{W}_{s} \mathbf{s}_{s}[n],
	\end{equation}  
	where $\mathbf{s}_{c}[n] \in \mathbb{C}^{K \times 1 } $ denotes the $n$-th column of $\mathbf{S}_c $ and $\mathbf{s}_{s}[n] \in \mathbb{C}^{M \times 1 } $ denotes the $n$-th column of $\mathbf{S}_s $. 
	Without loss of generality, we assume that the sensing waveforms satisfy $(1 / N) \mathbb{E}\left\{\sum_{n=1}^{N} \mathbf{s}_{s}[n] \mathbf{s}_{s}^{H}[n]\right\}=\mathbf{I}_{M}$.
	Similarly, the communication symbols are a generalized smooth random process with a unit average power $(1 / N) \mathbb{E}\left\{\sum_{n=1}^{N} \mathbf{s}_{c}[n] \mathbf{s}_{c}^{H}[n]\right\}=\mathbf{I}_{K}$. 
	Moreover, the communication symbols are uncorrelated with the sensing waveforms, i.e., $(1 / N) \mathbb{E}\left\{\sum_{n=1}^{N} \mathbf{s}_{s}[n] \mathbf{s}_{c}^{H}[n]\right\} = \mathbf{0}_{M\times K}$\cite{liu2020joint}. }

	Notice that waveform design and beamforming matrix design have equivalent effects, with the only difference being that waveform design cannot separately acquire beamforming matrix information for sensing and communication purposes. Our goal is to design the matrices $\mathbf{W}_{c}$ and $\mathbf{W}_{s}$ to implement transmit beamforming. 
	
	\subsection{Multi-user Communication Performance Metric}
	We consider a downlink multi-user MIMO-ISAC system with $K$ single antenna users, assuming that the downlink communication channel is a narrowband Rayleigh fading channel and is constant within a communication frame. 
	\textcolor{black}{  The channel output signal received by the $k$-th user at time index $n$ can be expressed as 
	\begin{equation}
		\begin{aligned}
			{y}_{k}[n] & = \underbrace{\mathbf{h}_{k}^H \mathbf{w}_{c, k} {s}_{c, k}[n]}_{\text {Desired Signal }} +\underbrace{\sum_{i=1, i \neq k}^{K} \mathbf{h}_{k}^H \mathbf{w}_{c, i} {s}_{c, i}[n]}_{\text {MUI }} \\
			& +\underbrace{\sum_{i=1}^{M} \mathbf{h}_{k}^H \mathbf{w}_{s, i} {s}_{s, i}[n]}_{\text {SI }}+\underbrace{{z}_{c, k}[n]}_{\text {Noise }}, 
		\end{aligned} 
	\end{equation} 
	where $\mathbf{h}_{k}\in\mathbb{C}^{ M\times 1}$ is the channel vector between the ISAC transmitter and the $k$-th user. ${\mathbf{w}_{c,k}\in\mathbb{C}^{ M\times 1}}$ is the $k$-th column of $\mathbf{W}_c$ and ${\mathbf{w}_{s,i}\in\mathbb{C}^{ M\times 1}}$ is the $i$-th column of $\mathbf{W}_s$.
	${s}_{c, k}[n]$ is the $k$-th element of $\mathbf{s}_{c}[n]$ and ${s}_{s, i}[n]$ is the $i$-th element of $\mathbf{s}_{s}[n]$.
    ${z}_{c,k}[n] \sim \mathcal{CN}\left( 0,{{\sigma }_k^{2}}\right) $ is additive white Gaussian noise  (AWGN) received at the $k$-th user. During a coherent processing interval, the channel remains approximately constant, allowing us to calculate the SINR of communication users based on the statistical expectation. The SINR at the $k$-th communication receiver can be expressed as 
	\begin{equation}
		\begin{aligned}
			\gamma_{k} &= \frac {\mathbb{E} \left \{\left|\mathbf{h}_{k}^H \mathbf{w}_{c,k}\right|^2 \right \}}  
			{ {\mathbb{E} \left\{  \sum_{i \neq k}  \left|\mathbf{h}_{k}^H \mathbf{w}_{c,i}\right|^{2} \right\} }+ \mathbb{E} \left \{ \sum_{i=1}^{M}  \left|\mathbf{h}_{k}^H \mathbf{w}_{s,i}\right|^{2}\right \} + \sigma_k^{2} }, 
		\end{aligned}
		\label{equa 3} 
	\end{equation} 
based on $(1 / N) \mathbb{E}\left\{\sum_{n=1}^{N} \left|{s}_{c,k}[n]\right|^2 \right\}=1$, $(1 / N) \mathbb{E}\left\{\sum_{n=1}^{N} \left| {s}_{s,i}[n]\right|^2\right\}=1$ and $(1 / N) \mathbb{E}\left\{\sum_{n=1}^{N} \left|{z}_{c,k}[n]\right|^2 \right\}=\sigma_k^2$. }

	The parameter $\gamma_{k}$ is a standard performance indicator to ensure reliable transmission of communication symbols.
	Specifically, when the sensing target and the communication users are nearby (less than one beamwidth), the SI will cause a reduction in the value of $\gamma_{k}$ at the communication receiver, decreasing the overall system capacity. Hence, (4) provides a practical depiction of the downlink propagation and incorporates the inherent trade-off between sensing and communication.
	
	\subsection{Sensing Performance Metric}
	The BS obtains the desired target motion information by transmitting several beams with high directivity from the array antenna. The presence of environmental targets can result in signal-dependent clutter, thereby affecting the sensing accuracy. To enhance the robustness of the ISAC system, this paper discusses the impact of steering vector mismatches caused by imperfect CSI or variations in environmental conditions.

	\textcolor{black}{ The baseband signal at time index $n$ for target sensing located with angle $\theta$ is denoted as 
		\begin{equation}
			{y}_{s,0}[n]=\mathbf{a}^H_t\left( {{\theta }} \right)\mathbf{x}[n], n=1,2,\ldots, N,
		\end{equation}
		where  ${\mathbf{a}_{t}}\left( \theta  \right)$ is the $M\times 1$} transmit steering vector containing complex-valued elements.
	For a uniform linear array (ULA),  the steering vector is given by
	\begin{equation}
		{\mathbf{a}_{t}}\left( \theta  \right)=\frac{1}{\sqrt{M}}{{\left[ 1,{{e}^{j2\pi \frac{{{d}_{T}}}{\lambda }\sin \theta }},\ldots ,{{e}^{j2\pi \frac{{{d}_{T}}}{\lambda }\left( M-1 \right)\sin \theta }} \right]}^{T}}, 
	\end{equation}
	where ${{d}_{T}}$ is the array inter-element spacing at the BS transmitter and $\lambda $ is the transmit signal wavelength.

\textcolor{black}{ Suppose that there is a target located at angle $\theta$ along with $P$ signal-dependent clutter sources located at $\theta_{p} \neq \theta_{0}, p=1,\ldots,P$.	
The received signal of the BS receiving array at time index $n$ is   
	\begin{align}\label{ys}
		\mathbf{y}_{s}[n]  = & {{\alpha }_{0}} {{\mathbf{a}}_{r}}\left( \theta_0  \right){{\mathbf{a}}^H_{t}}\left( \theta_0  \right) \mathbf{x}[n] \notag \\
		& + \sum\limits_{p=1}^{P}{{{\alpha }_{p}}
			{{\mathbf{a}}_{r}}\left( \theta_p  \right){{\mathbf{a}}^H_{t}}\left( \theta_p  \right)
			\mathbf{x}[n]}
		+\mathbf{{z}}_{s}[n],
	\end{align} 
where $\mathbf{y}_{s}[n] \in \mathbb{C}^{M\times 1}$,  $\alpha_{0}$ and $\alpha_{p}$ denote the complex amplitudes of the target and $p$-th clutter, indicating the radar cross section (RCS) and channel propagation effects.
The vector ${\mathbf{a}_{r}}\left( \theta \right)$ is the $M\times 1 $ steering vector  due to the propagation delays from the target to the receive elements, given by  
\begin{equation}
	{\mathbf{a}_{r}}\left( \theta  \right)=\frac{1}{\sqrt{M}}{{\left[ 1,{{e}^{j2\pi \frac{{{d}_{R}}}{\lambda }\sin \theta }},\ldots ,{{e}^{j2\pi \frac{{{d}_{R}}}{\lambda }\left( M-1 \right)\sin \theta }} \right]}^{T}}, 
\end{equation}
  where ${{d}_{R}}$ is the array inter-element spacing at the BS receiver and $\mathbf{z}_{s}[n]$ is distributed as complex AWGN under distribution ${{\mathbf{z}}_{s}[n]}\sim \mathcal{CN}\left( 0,{{\sigma }^{2}}{\mathbf{I}_{M}} \right)$. 
  To simplify, we remove the time index and define $\mathbf{A}(\theta) = {{\mathbf{a}}_{r}}\left( \theta  \right){{\mathbf{a}}^H_{t}}\left( \theta  \right)$, which is dependent on the angle $\theta$. (7) can then be reformulated as:  
	\begin{equation}\label{sening receiving}
		\begin{aligned}
			\mathbf{y}_{s} &=\underbrace{\alpha_{0} \mathbf{A}\left(\theta_{0}\right) \mathbf{x}}_{\text {Echo }} \\
			& +\underbrace{\sum_{p=1}^{P} \alpha_{p} \mathbf{A}\left(\theta_{p}\right) \mathbf{x}}_{\text {Clutter }}+\underbrace{\mathbf{z}_{s}}_{\text {Noise }}.
		\end{aligned}
	\end{equation}}	
	
	To mitigate clutter and maximize the sensing output SCNR, the received signal is subjected to filtering using the weight vector $\boldsymbol{\omega} \in \mathbb{C}^{M\times 1} $. Consequently, the sensing output signal, denoted as ${z}_{out}$, can be given by
	\begin{equation}
		{z}_{out}=\boldsymbol{\omega}^{H} \mathbf{y}_{s}.
	\end{equation}
	
	Thus, the output SCNR can be formulated as
	\textcolor{black}{
	\begin{equation}
		\begin{aligned}
			\rho & =\operatorname{SCNR}(\mathbf{W}, \boldsymbol{\omega}) \\
			& =\frac{\mathbb{E}\left\{\left|\alpha_{0}\right|^{2}\left|\boldsymbol{\omega}^{H} \mathbf{A}\left(\theta_{0}\right) \mathbf{x}\right|^{2}\right\}} 
			{\mathbb{E}\left\{ \sum_{p=1}^{P}\left| \alpha_{p} \boldsymbol{\omega}^{H}   \mathbf{A}\left(\theta_{p}\right) \mathbf{x}\right|^{2} \right\}
			+ \mathbb{E} \left\{\sigma^{2} \boldsymbol{\omega}^H \boldsymbol{\omega}\right\} }\\
			& =\frac{\boldsymbol{\omega}^{H} \Phi(\mathbf{x}) \boldsymbol{\omega}}{\boldsymbol{\omega}^{H} \Phi_{\mathrm{cn}}(\mathbf{x}) \boldsymbol{\omega}},
		\end{aligned} \label{equa 10} 
	\end{equation} 
 where $ \Phi(\mathbf{x})=\left|\alpha_{0}\right|^{2} \mathbf{A}(\theta_{0}) \mathbf{x} \mathbf{x}^{H} {\mathbf{A}}^{H} ({\theta }_{0}) $.  }
 $ \Phi_{\mathrm{cn}}(\mathbf{x})=\mathbf{R}_{c}+\sigma^{2} \mathbf{I}$ is the total interference covariance matrix (in this case, clutter plus noise), and the clutter covariance matrix ${\mathbf{R}}_{c}$ is given by 
	\begin{equation}\label{equa 8}
		{{\mathbf{R}}_{c}}=\mathbb{E} \left\{ \sum\limits_{p=1}^{P}{{{\left| {{\alpha }_{p}} \right|}^{2}}\mathbf{A}\left( {{\theta }_{p}} \right)\mathbf{x} {{\mathbf{x}}^{H}}{{\mathbf{A}}^{H}}\left( {{\theta }_{p}} \right) } \right\}.
	\end{equation}
	
	However, accurate acquisition of (\ref{equa 8}) is often challenging because the motion of clutter sources or vibrations in the antenna platform can cause rapid channel variations\cite{qian2017null, 10036139 }, leading to errors in the clutter steering vectors, which not only directly impact the clutter covariance matrix but also result in a mismatch between the adaptive receive filter and the received signal. Consequently, effective clutter suppression becomes challenging. Considering the limited feasibility of continuously updating the adaptive receive filter, enhancing the robustness of the ISAC transceiver design is imperative to ensure optimal performance.	

\textcolor{black}{ 
	Assume that there exist directional perturbation in the clutter received by the BS, resulting in a phase disturbance in the array response vector. 
	Therefore, we consider the mismatched steering vector resulting from phase disturbances as the following vector random process. }
\begin{equation}
	\tilde{\mathbf{e}}\triangleq {{\left[ 1, {{e}^{j\varphi }}, \ldots, {{e}^{j\left( M-1 \right)\varphi }} \right]}^{T}},
\end{equation}
where $\varphi$ is a zero-mean distribution random variable and  the covariance matrix of $\tilde{\mathbf{e}}$ is then given by
\begin{equation}
	\mathbf{T}=\operatorname{cov}\left( \tilde{\mathbf{e}} \right).
\end{equation} 

\textcolor{black}{Due to the existence of $\tilde{\mathbf{e}}$, the clutter-plus-noise vector becomes $ \mathbf{y}_\text{CI} \circ \tilde{\mathbf{e}} $, where $\mathbf{y}_\text{CI}$ represents the latter two terms in (\ref{sening receiving}) and $\circ$ represents the Hadamard product. 
	As a result, there exists a relationship for the covariance matrix of interference by $ \rm cov(\mathbf{y}_\text{CI} \circ \tilde{\mathbf{e}})= {{\Phi }_{\text{cn}}} \circ \mathbf{T} $.} 
In other words, when the Hadamard product is applied to two uncorrelated conformal vector random processes, the resulting additive covariance matrix equals the Hadamard product of their respective covariance matrices \cite{guerci1999theory}. In consequence, for a matrix ${{\Phi }_{\text{cn}}}(\mathbf{x})$, the tapered matrix is given by ${{\tilde{\Phi }}_{\text{cn}}}$, i.e.,
	\begin{equation}
		{{\tilde{\Phi }}_{\text{cn}}}={{\Phi }_{\text{cn}}}\circ \mathbf{T}.
	\end{equation}
	Hence, the sensing output SCNR can be redefined as
	\begin{equation}
		\rho=\frac{\boldsymbol{\omega}^{H} \Phi(\mathbf{x}) \boldsymbol{\omega}}{\boldsymbol{\omega}^{H} \tilde{\Phi }_{\mathrm{cn}}(\mathbf{x}) \boldsymbol{\omega}}.
	\end{equation} 
Moreover, it is assumed that any possible self-interference arising from the transmit and receive arrays is effectively mitigated by employing an appropriate method \cite{qian2017null }.
	
	\section{Problem Formulation}
	Based on the above analysis, this paper proposes an efficient method for the MIMO-ISAC system in a cluttered environment, considering the effects of MUI and SI. Specifically, the objective is to maximize the SCNR at the BS receiver while incorporating extended constraints and ensuring robustness by considering mismatched steering vectors.

In practical sensing systems, utilizing constant-modulus signals in amplifiers is crucial for achieving maximum efficiency and minimizing excessive amplitude modulation. \textcolor{black}{ It is assumed that the transmitted energy is normalized to a standard level. Given that $\mathbf{R}$ represents the covariance matrix of the transmit ISAC signal, $\mathbf{R}$ can be expressed jointly as
\begin{equation}	
	\begin{aligned}
		\mathbf{R} & =\mathbb{E}\left[\mathbf{x} \mathbf{x}^H\right] = \mathbf{W} \mathbf{W}^H. \\
	\end{aligned}
\end{equation}
The per-antenna power constraint implies that for each antenna $m$, it is required that
	\begin{equation}
		{\left| \mathbf{R}_{m,m} \right|} \leq {{{P}_{t}}}/{M}, \ \forall m,
	\end{equation}
	where $\mathbf{R}{m,m}$ denotes the $m$-th diagonal element of $\mathbf{R}$ and $P_{t}$ is the total available transmit power. This constraint ensures that the power transmitted from each antenna does not exceed a specified threshold, which is crucial for effective signal transmission within a coherent processing interval (CPI).}
	
	\textcolor{black}{The accuracy of target detection, identification, and tracking for ISAC signals relies on the similarity between the covariance of the optimized transmitted signal and the covariance of the reference signal, denoted as $\mathbf{R}_{0}$. The error between ${{\mathbf{R}}}$ and ${{\mathbf{R}}_{0}}$ is defined as a constraint that indicates the sensing performance, expressed as 
	\begin{equation}
		\begin{aligned}
			L(\mathbf{R})  =\left\|\mathbf{R}_0-\mathbf{R}\right\|_F,
		\end{aligned}
	\end{equation}
	where the notation $\|\cdot\|_F$ represents the Frobenius norm.} 
	
	To meet the requirements of sensing, the maximum value of the covariance error needs to be smaller than $\xi $, i.e., $L\left( \mathbf{R} \right)\le \xi $, where $\xi \in \left[ 0,\alpha {{P}_{t}} \right]$ and $\alpha $ is a real number between 0 and 2 \cite{ cui2013mimo, karbasi2015robust, he2020joint}. By decreasing the similarity threshold $\xi $, the soft control on the modulus variolation, peak sidelobe level, and range resolution of the ISAC signal can be obtained.

In the presence of communication requirements, allocating power for transmitting communication symbols leads to a degradation in sensing functionality. 
To ensure the minimum level of communication QoS for each user, it is required to ensure that the SINR at each communication receiver is higher than a given threshold $\Gamma $, i.e.,
\begin{equation}
	{{\gamma }_{k}}\ge \Gamma,  \ k=1,\ldots ,K. 
\end{equation}

Based on the above analysis, we propose a joint optimization approach to maximize the SCNR at the sensing receiver by jointly optimizing both transmit beamforming and receive filter, which can be achieved by solving the following optimization problem. 
\begin{subequations}
	\begin{align}
		\max _{\boldsymbol{\omega}, \mathbf{W}}  & \ \rho \notag\\
		\text { s.t. } &\mathbf{R}=\mathbf{W} \mathbf{W}^H, \mathbf{R} \in \mathcal{S}_{M}^{+}, \\
		& \textcolor{black}{ {\left| \mathbf{R}_{m,m} \right|} \leq {{{P}_{t}}}/{M}, \ \forall m,}             \\
		& L(\mathbf{R}) \leq \xi, \\
		&\gamma_{k} \geq \Gamma, \ k=1, \ldots, K, 
	\end{align}
\end{subequations} \label{p19} 
\textcolor{black}{where $\mathcal{S}_{M}^{+}$ represents a set of positive semidefinite matrices of dimension $M$.}
	
	The optimization problem (21) is challenging due to the non-convexity introduced by the quadratic equality constraint in (21a). Since decision variables are coupled in the fractional form constraint (21d), problem (21) is challenging to be solved directly. Nonetheless, the resultant optimization problem can be efficiently solved when one of $\mathbf{W}$ and $\boldsymbol{\omega}$ is fixed. 
This motivates us to use an alternating optimization method that involves two iterative steps.	
\begin{enumerate}
	\item \textit{Receive Filter Design:} 
	The receive filter $\boldsymbol{\omega}$ optimization subproblem can be reduced to an MVDR problem for a given transmit beamforming. The difference is that the mismatched steering vectors are considered in this paper, so the interference signal matrix needs to be reconstructed. Therefore, this optimization problem first requires a modification of the optimization objective. Then, a filter design with enhanced robustness is obtained at the receiver.
	
	\item \textit{Transmit Beamforming Design:}
	Under the optimized receive filter in 1), the overall transmit beamforming $\mathbf{W}$ optimization subproblem is a non-convex fractional problem. Due to the different dimensions of sensing beamforming and communication beamforming, it is impractical to obtain both $\mathbf{W}_c$ and $\mathbf{W}_s$ simultaneously. Therefore, the quadratic term of the optimized communication beamforming is first optimized, and then the sensing beamforming is obtained by decomposition.
	
\end{enumerate}

The optimal solution is achieved through iterative optimization of the variables until convergence. The interference management effect becomes a closed loop at the BS transceiver through joint optimization. 
	
	\begin{algorithm}[t] 
		\caption{Robust Joint Transmit Beamforming and Receive Filter Design Algorithm}\label{algorithm}
		\KwIn{$\mathbf{A}\left(\theta_{0}\right)$, $\mathbf{A}\left(\theta_{p}\right)$, $\alpha_{0}$, $\alpha_{p}$, $\forall p$, $P_{t}$, $\mathbf{x}_{0}$, $\mathbf{h}_k$, $\sigma_k^2$, $\forall k$, $\sigma^2$, $\Gamma$, $\alpha$, $\Delta $, $\vartheta$.}
		\KwOut {$\boldsymbol{\omega}^{opt}, \mathbf{W}^{opt}$.}
		Initialize $\mathbf{R}^{(0)},\mathbf{R}_1^{(0)}, \ldots, \mathbf{R}_{K}^{(0)}, n=0$.\\ 
		
		\While{no convergence}{  
			Set $n=n+1$.\\
			
			Get the “tapered matrix” ${{\tilde{\Phi }}_{\text{cn}}}$  by (23).
			
			Calculate the optimal solution $\boldsymbol{\omega}^{(n+1)}$ by (25).\\
			
			\For {$i=1:iterations$}{Calculate auxiliary variable ${y}^{(n+1)}$ by (30).\\	
				
				Update the optimal $\hat{\mathbf{R}}^{(0)}, \hat{\mathbf{R}}_1^{(0)}, \ldots, \hat{\mathbf{R}}_{K}^{(0)}$ by solving the optimization problem (33).	
			}	
		}	
		Compute ${\tilde{\mathbf{w}}_{1}}, \ldots , {\tilde{\mathbf{w}}_{K}}$ via (37).
		
		Let $\mathbf{W}_c =[\tilde{\mathbf{w}}_{1}, \ldots , \tilde{\mathbf{w}}_{K} ]$.
		
		Compute ${\tilde{\mathbf{w}}_{K+1}},\ldots ,{\tilde{\mathbf{w}}_{K+M}}$ via (38).
		
		Let $\mathbf{W}_s =[\tilde{\mathbf{w}}_{K+1}, \ldots , \tilde{\mathbf{w}}_{K+M} ]$.
		
		Return the receive filter $\boldsymbol{\omega}^{opt}$ and overall beamforming matrix $\mathbf{W}^{opt}=[\mathbf{W}_c,\mathbf{W}_s].$
	\end{algorithm}
	 
	\section {Proposed Solutions for Robust Transceiver Design}
	In this section, we first tackle the robust receive filter design subproblem by considering the mismatched steering vector in Section \uppercase\expandafter{\romannumeral4}-A. 
	\textcolor{black}{In Section \uppercase\expandafter{\romannumeral4}-B, we introduce a subproblem for transmit beamforming and derive the communication beamforming and sensing beamforming matrices separately.} 

\subsection{The CMT-based Receive Filter Design}
\textcolor{black}{The stable covariance characteristics of fixed clutter in the environment simplify its estimation and elimination. However, mismatches between adaptive filters and signals pose a significant challenge in MIMO-ISAC systems \cite{qian2017null}. Expanding the nulling width to enhance the robustness of the adapted beampattern is a desirable goal.} 

To modify the clutter covariance matrix, the CMT technique introduced in \cite{guerci1999theory} is employed. Specifically, we use the MZ taper defined as follows.  
	\begin{equation}
		\mathbf{T}_{MZ}={{\left[ {{a}_{mn}} \right]}_{M\times M}}=\left\{ \operatorname{sinc}\left( \left( m-n \right)\vartriangle  \right) \right\}.
	\end{equation}
	$\mathbf{T}_{MZ}$ represents a matrix of size $M\times M$, where the $\left( m,n \right)$-th element is $\operatorname{sinc}\left( \left( m-n \right)\vartriangle  \right)$. The symbol $\Delta$ represents the dithering range, wherein the signal can be effectively suppressed and attenuated in the output through the design of the receiving filter, thereby forming broader nulls.
The choice of $\Delta$ is crucial to ensure that the taper operation effectively attenuates signals arriving at angles between adjacent clutter points. By increasing the value of $\Delta$, the null width formed in the clutter direction becomes wider, indicating a higher tolerance of mismatches. However, an excessively widened null region can adversely affect the SCNR at the sensing receiver, as discussed in the simulation results provided in Section \uppercase\expandafter{\romannumeral5}-D.  
	
	By using $\mathbf{T}_{MZ}$, the ``tapered matrix" is reformulated by 
\begin{equation}\label{equa21}
	{{\tilde{\Phi }}_{\text{cn}}}={{\Phi }_{\text{cn}}}\circ \mathbf{T}_{MZ}.
\end{equation}
As pointed out in \cite{guerci1999theory }, the MZ taper is equivalent to introducing a uniformly distributed	coherent phase dither. \textcolor{black}{ Hence, the power of each observation point can be adjusted within its neighborhood, which enables an increased nulling width. This feature effectively suppresses clutter even with errors in clutter angle directions.}

For a fixed transmit beamforming, the receive filter optimization subproblem can be reduced to 
\begin{equation}
	\underset{\boldsymbol{\omega }}{\mathop{\max }}\,\frac{{{\boldsymbol{\omega }}^{H}}\Phi (\mathbf{W})\boldsymbol{\omega }}{{{\boldsymbol{\omega }}^{H}} {{{\tilde{\Phi }}}_{\text{cn}}}(\mathbf{W})\boldsymbol{\omega }},
\end{equation} 
which refers to the MVDR problem. The design principle of this filter aims to transmit the signal of interest in the desired direction without distortion while minimizing the variance of the beamforming output noise, which is a generalized noise containing interference \cite{cox1987robust }. The optimal solution has the following form
\begin{equation}
	\boldsymbol{\omega }=f\left( {{\left( {{{\tilde{\Phi }}}_{\text{cn}}} \right)}^{-1}}\Phi  \right),
\end{equation}
where $ f (\cdot)$ defines the largest eigenvector corresponding to the generalized eigenvalue of the matrix.
This method utilizes the inverse of the sample covariance matrix for filter design. 
The CMT technique offers inherent low complexity by utilizing a single matrix Hadamard product and enhances the robustness, which has been successfully applied in adaptive beamforming \cite{guerci1999theory,zhang2015interference, ollila2022regularized}.

	\subsection{The SDR-QT based Transmit Beamforming Design}	
	The subproblem for a fixed linear receive filter is transformed into a non-convex fractional quadratically constrained quadratic programming (QCQP) problem. In this subsection,  the SDR approach  \cite{zhang2000quadratic } and QT approach \cite{ shen2018fractional,nesterov1998semidefinite} are employed to tackle the subproblem.

Since ${\mathbf{w}_{i}}$ is the $i$-th column of $\mathbf{W}$, then
\begin{equation}
	\mathbf{R}=\sum\limits_{i=1}^{M+K}{{{\mathbf{w}}_{i}}\mathbf{w}_{i}^{H}}=\sum\limits_{i=1}^{M+K}{{{\mathbf{R}}_{i}}},
\end{equation} 
where ${\mathbf{R}}_{i}\in\mathbb{C}^{M \times M}$ with $i\in \left[ 1, K+M \right]$ satisfies $\text{rank}\left( {{\mathbf{R}}_{i}} \right)=1$, i.e., ${\mathbf{R}}_{i}\in \mathcal{S}_{M}^{+}$. 
Accordingly, the optimization of $\mathbf{W}$ in problem (21) is equivalent to optimizing $\mathbf{R}$, which is further obtained through the decomposition.

The $k$-th row of $\mathbf{H}$, the QoS constraints of communication users can be transformed into linear form as follows.
	\textcolor{black}{
	\begin{equation}
		\begin{aligned}
			\gamma_k &=\frac{ \mathbf{h}_{k}^{H} \mathbf{w}_{k} \mathbf{w}_{k}^{H}  \mathbf{h}_{k}}{\sum_{i \neq k} \mathbf{h}_{k}^{H} \mathbf{w}_{i} \mathbf{w}_{i}^{H}  \mathbf{h}_{k}+\sigma_k^{2}} \\
			& =\frac{\mathbf{h}_k^H \mathbf{R}_k \mathbf{h}_k}{\sum_{1 \leq i \leq M+K, i \neq k} \mathbf{h}_k^H \mathbf{R}_i \mathbf{h}_k+\sigma_k^2} \\
			& =\frac{\mathbf{h}_k^H \mathbf{R}_k \mathbf{h}_k}{\mathbf{h}_k^H \mathbf{R }\mathbf{h}_k-\mathbf{h}_k^H 
				\mathbf{R}_k \mathbf{h}_k+\sigma_k^2} \geq \Gamma .
		\end{aligned}
	\end{equation}}
	By simple rearrangement, the communication constraint can be obtained as follows.
	\begin{equation}
	\textcolor{black}{	\left( 1+{{\Gamma }^{-1}} \right)\mathbf{h}_{k}^{H}{{\mathbf{R}}_{k}}\mathbf{h}\ge \mathbf{h}_{k}^{H}\mathbf{R}{\mathbf{h}_{k}}+{{\sigma }_k^{2}}, \ k=1,\ldots ,K.}
	\end{equation}
	
	Since the objective function is essentially a fractional optimization problem with a concave numerator and a convex quadratic term in the denominator, it is no longer a simple single-ratio optimization problem. Therefore, traditional fractional optimization algorithms such as Charnes-Cooper transform, and Dinkelbach's transform cannot guarantee that the optimal value of the transformed objective function is the same as that of the original fractional optimization problem\cite{shen2018fractional }. \textcolor{black}{In this regard, an iterative optimization methodology proves effective by alternately optimizing the original variables $\mathbf{R}$ and the auxiliary variable $y$. Thus, the optimization subproblem can be redefined by utilizing the following lemma of }
	\begin{subequations}\label{p27}
		\begin{align}
			\max_{\mathbf{R},\{\mathbf{R}_i\},y } & 2y\sqrt{{{\boldsymbol{\omega }}^{H}}\Phi (\mathbf{R})\boldsymbol{\omega }}-{{y}^{2}}{{\boldsymbol{\omega }}^{H}}{{\tilde{\Phi }}_{\text{cn}}}(\mathbf{R})\boldsymbol{\omega } \notag\\
			\text { s.t.} \ & \mathbf{R}=\sum_{i=1}^{M+K} \mathbf{R}_i,  \\
			&  \mathbf{R}_i \in \mathcal{S}_{M}^{+}, \   \mathrm{rank}\left(\mathbf{R}_i\right)=1, \ i=1, \ldots, K+M,  \\
			& \textcolor{black}{{\left| \mathbf{R}_{m,m} \right|} \leq {{{P}_{t}}}/{M}, \ \forall m, }\\
			& L(\mathbf{R}) \leq \xi, \\
			&\textcolor{black}{\left(1+\Gamma^{-1}\right) \mathbf{h}_k^H \mathbf{R}_k \mathbf{h} \geq \mathbf{h}_k^H \mathbf{R} \mathbf{h}_k+\sigma_k^2, \ \forall k.}
		\end{align}
	\end{subequations}

	\renewcommand{\qedsymbol}{\ensuremath{\blacksquare}}
	\newtheorem{lemma}{Lemma}
	\newenvironment{mylemma}[1][Lemma]{\par
		\normalfont\textit{#1}\quad}{\qed\par}
	
	\begin{lemma}
		\textit{ (Optimality Quadratic Transform) } 
		\upshape When $\mathbf{R}$ and $\boldsymbol{\omega }$ are held fixed, the optimal $y$ can be updated in the following closed form expression
		\begin{equation}
			{{y}^{*}}=\frac{\sqrt{{{\boldsymbol{\omega }}^{H}}\Phi (\mathbf{R})\boldsymbol{\omega }}}{{{\boldsymbol{\omega }}^{H}}{{{\tilde{\Phi }}}_{\text{cn}}}(\mathbf{R})\boldsymbol{\omega }}.
		\end{equation}
		
		The original fractional optimization objective function can be reformulated as follows according to QT,
		\begin{equation}\label{equa27}
			\textcolor{black}{g} \left( \mathbf{R},y \right)\ = 2y\sqrt{{{\boldsymbol{\omega }}^{H}}\Phi (\mathbf{R})\boldsymbol{\omega }}-{{y}^{2}}{{\boldsymbol{\omega }}^{H}}{{\tilde{\Phi }}_{\text{cn}}}(\mathbf{R})\boldsymbol{\omega },
		\end{equation}
		which is concave in $\mathbf{R}$ for a fixed $y$. The optimal solution to the problem (29) converges to the global optimal solution to the problem (21).
	\end{lemma}
	
	\renewenvironment{proof}[1][\proofname]{\par
		\normalfont\textit{#1}\quad}{\qed\par}
	\begin{proof}   
		With fixed vectors $y$ and $\boldsymbol{\omega }$, where $\Phi (\mathbf{R})$ is a concave function and ${{\tilde{\Phi }}_{\text{cn}}}(\mathbf{R})$ is a convex function. The square root function is a monotonically increasing concave function, and the restated objective function becomes a concave maximization problem over $\mathbf{R}$. The optimal $\mathbf{R}$ can be efficiently obtained through numerical convex optimization. By iteratively updating $y$, the algorithm converges to the global optimal solution for the problem (21).
	Readers could refer to \cite{shen2018fractional} for detailed proof.
	\end{proof}
	
	The problem (29) is still non-convex due to the rank-one constraint (29b). Thus, SDR can be adopted to relax this constraint, and the problem (29) can be relaxed as
	\begin{subequations}\label{p30}
		\begin{align}
			\max_{\mathbf{R}, \{\mathbf{R}_i\},y } & \textcolor{black}{g} (\mathbf{R}, y) \notag\\
			\text { s.t.} \ & \mathbf{R}=\sum_{i=1}^{M+K} \mathbf{R}_i, \\ 
			& \mathbf{R}_i \in \mathcal{S}_{M}^{+}, \ i=1, \ldots, K+M, \\
			& \textcolor{black}{{\left| \mathbf{R}_{m,m} \right|} \leq {{{P}_{t}}}/{M}, \ \forall m, }\\
			& L(\mathbf{R}) \leq \xi, \\
			& \textcolor{black}{\left(1+\Gamma^{-1}\right) \mathbf{h}_k^H \mathbf{R}_k \mathbf{h} \geq \mathbf{h}_k^H \mathbf{R} \mathbf{h}_k+\sigma_k^2, \ \forall k.}
		\end{align}
	\end{subequations}
	
	Since the dimensions of $\mathbf{W}, {{\mathbf{W}}_{c}},{{\mathbf{W}}_{s}}$ are different, it is difficult to handle even if the rank-one constraint is ignored. A feasible solution is to perform matrix decomposition on $\{{\mathbf{R}}_{k}\}$ with $k=1,\ldots, K$ and $\mathbf{R}-\sum\limits_{i=1}^{K}{{{\mathbf{R}}_{i}}}$ separately, where the former obtains the communication beamforming matrix and the latter obtains the sensing beamforming matrix \cite{ liu2020joint }. Temporarily ignoring the influence of variable ${{\left\{ {{\mathbf{R}}_{i}} \right\}}_{i\ge K+1}}$ in optimizing $\{{\mathbf{R}}_{k}\}$ can reduce the number of optimization variables without losing communication performance constraints. i.e.,
	\begin{subequations}\label{p30}
		\begin{align}
			\max _{\mathbf{R}, \mathbf{R}_1, \ldots \mathbf{R}_K, y} & \textcolor{black}{g} (\mathbf{R}, y)  \notag\\
			\text { s.t. } & \mathbf{R} \in \mathcal{S}_{M}^{+}, \ \mathbf{R}-\sum_{k=1}^K \mathbf{R}_k \in \mathcal{S}_{M}^{+} \\
			& \mathbf{R}_k \in \mathcal{S}_{M}^{+}, \ k=1, \ldots, K, \\
			& \textcolor{black}{{\left| \mathbf{R}_{m,m} \right|} \leq {{{P}_{t}}}/{M}, \ \forall m, }\\
			&  L(\mathbf{R}) \leq \xi, \\
			&\textcolor{black}{\left(1+\Gamma^{-1}\right) \mathbf{h}_k^H \mathbf{R}_k \mathbf{h} \geq \mathbf{h}_k^H \mathbf{R} \mathbf{h}_k+\sigma_k^2, \ \forall k.}
		\end{align}
	\end{subequations}
	
	Thus, the number of optimization variables has decreased to $K$+1. 
Problem (33) is known as an SDR of (29) because (33) is an instance of semidefinite programming (SDP) and can be solved using a generic SDP solver\cite{2012On}. The relaxed optimization model (33) is a convex QCQP because the objective function is a positive-semidefinite quadratic form, and all the constraints are either linear or semidefinite. 
Then the global solution of $\tilde{\mathbf{R}},{{\tilde{\mathbf{R}}}_{1}},\ldots {{\tilde{\mathbf{R}}}_{K}}$ can be obtained in polynomial time with convex optimization toolboxes. The tightness of SDR for the problem (33) can be guaranteed according to the following lemma.

\begin{lemma}
	\textit{ (Tightness of SDR)}
	\upshape After applying SDR to the problem (29), the problem (33) is convex, and the optimal solution 	$\hat{\mathbf{R}},{{\hat{\mathbf{R}}}_{1}},\ldots {{\hat{\mathbf{R}}}_{K}}$ can be used to find a rank-one optimal solution $\tilde{\mathbf{R}},{{\tilde{\mathbf{R}}}_{1}},\ldots {{\tilde{\mathbf{R}}}_{K}}$. 
\end{lemma}

\begin{proof}
	\textcolor{black}{After obtaining the optimal solution
	$\hat{\mathbf{R}},{{\hat{\mathbf{R}}}_{1}}, \ldots,  {{\hat{\mathbf{R}}}_{K}}$, the covariance matrix of the transmit signal satisfies $\tilde{\mathbf{R}} = \hat{\mathbf{R}}$ without the rank-one constraint. Therefore, it is necessary to demonstrate that the covariance matrix of the communication signals in $\tilde{\mathbf{R}}_k$ satisfies all constraints outlined in problem (33). } 
\textcolor{black}{	By constructing }
	\begin{equation}
		{{\mathbf{\tilde{R}}}_{k}}=\frac{{{{\mathbf{\hat{R}}}}_{k}}{{\mathbf{h}}_{k}}\mathbf{h}_{k}^{H}\mathbf{\hat{R}}_{k}^{H}}{\mathbf{h}_{k}^{H}{{{\mathbf{\hat{R}}}}_{k}}{{\mathbf{h}}_{k}}}, \ \forall k,	
	\end{equation}
	\textcolor{black}{we can derive that }
	\begin{equation}
		\tilde{\mathbf{R}}_k = \tilde{\mathbf{w}}_k \tilde{\mathbf{w}}_k^H = \hat{\mathbf{R}}_k. 
	\end{equation}
	By substituting (35) into (33a) and (33e), we can obtain 
	\begin{equation}
		\left(1+\Gamma^{-1}\right) \mathbf{h}_k^H \tilde{\mathbf{R}}_k \mathbf{h} \geq \mathbf{h}_k^H \tilde{\mathbf{R}} \mathbf{h}_k+\sigma^2, \ \forall k,
	\end{equation}
	where $\tilde{\mathbf{R}}_k \in \mathcal{S}_{M}^{+}$ and $\tilde{\mathbf{R}}-\sum_{k=1}^K \tilde{\mathbf{R}}_k \in \mathcal{S}_{M}^{+}$. 
	\textcolor{black}{This indicates that $\tilde{\mathbf{R}}_k$ is also an optimal solution to the problem (33), which completes the proof [14]. }
\end{proof}	
	
	\textcolor{black}{
	Then, ${\tilde{\mathbf{w}}_{1}},\ldots ,{\tilde{\mathbf{w}}_{K}}$ can be derived due to (34) as 
	\begin{equation}
		{\tilde{\mathbf{w}}_{k}}={{\left( \mathbf{h}_{k}^{H}{\hat{\mathbf{R}}_{k}}{\mathbf{h}_{k}} \right)}^{-{1}/{2}\;}}{\hat{\mathbf{R}}_{k}}{\mathbf{h}_{k}},
	\end{equation}
	with $k = 1,\ldots ,K$.}  
Moreover, matrices $\{\tilde{\mathbf{R}}_i\}_{i\geqslant {K +1}}$ are calculated by the Cholesky decomposition \cite{ zhang2017matrix} according to the following relationship
	\begin{equation}\label{32}
		{{\mathbf{W}}_{s}}{{\mathbf{W}}_{s}}^{H}=\tilde{\mathbf{R}}-\sum\limits_{k=1}^{K}{{\tilde{\mathbf{w}}_{k}}\tilde{\mathbf{w}}_{k}^{H}}.
	\end{equation} 
	In other words, the $\mathbf {W}_c$ and $\mathbf {W}_s$ can be obtained by (37) and (38), respectively.	
	Joint optimization offers the dual advantages of interference suppression at both the transmitter and receiver compared to separate optimization.

\subsection{Algorithm Summary and Computational Complexity Analysis}
\textcolor{black}{ The robust joint transmit beamforming and receive filter algorithm is summarized in Algorithm 1, where we terminate the algorithm if
	\begin{equation}
		\frac{|{{\rho }^{i}}-{{\rho }^{i-1}}|} {{{\rho }^{i-1}}}<\vartheta,
	\end{equation}
	where $\vartheta$ is predefined small values (e.g., $10^{-3}$).	 
	The computational complexity of Algorithm 1 is determined by the number of iterations and the complexity at each iteration. 
	We present the computational complexity for each iteration in Table \ref{complexity}, where $N_O$ and $N_I$ denote the number of outer and inner iterations for the proposed algorithm to reach convergence, respectively. }
 	 
\setlength{\extrarowheight}{5.6pt}
\begin{table*}[htbp]\color{black}
	\centering
	\caption { Computation Complexity Analysis}
	\label{complexity}
	\begin{tabular}{|c|p{3.5cm}|p{6cm}|}
		\hline
	    & \textbf{Computation}& \textbf{Complexity}  \\ \hline			
		\multirow {3} {*} {Update $\mathbf{w}$ } 
		& $\Phi$              &  $\mathcal{O}\left( {{N}^{2}}{{M}^{2}} \right)$    \\ \cline{2-3}
		& ${{\left( {{{\tilde{\Phi }}}_{\text{cn}}} \right)}^{-1}}$ & $\mathcal{O}\left( {{N}^{3}}\left( {{M}^{2}}+{{M}^{3}} \right) \right)$ \\ \cline{2-3} 
		& $\mathbf{w}$        & $\mathcal{O}\left( {{N}^{2}}{{M}^{2}} \right)$   \\ \hline 
		
		\multirow {2} {*} {Update $\mathbf{W}$}
		&  $y^*$              &   $\mathcal{O}\left( {{M}^{2}} \right)$ \\ \cline{2-3} 
		& $\mathbf{W}$        &  $\mathcal{O}\left( K\left( NM+{{\left( NM \right)}^{2}}+{{\left( NM \right)}^{3}} \right) \right)$ \\ \hline
		
		\multirow{1}{*}{Total}               
		& \multicolumn{2}{c|} %
		{$\mathcal{O}(N^3 M^3)$ + $\mathcal{O}\left( N_O N_I \left( K\left( N M + \left( N M \right)^2 + \left( N M \right)^3 \right) \right) \right)$} \\ \hline

	\end{tabular}
\end{table*}
	 
	\section{Simulation Results}	
	In this section, we numerically evaluate the performance of the proposed joint transceiver design methods. The scenario in simulation is depicted in Fig. 1. Firstly, Section \uppercase\expandafter{\romannumeral5}-A summarizes the simulation settings. Then, in Section \uppercase\expandafter{\romannumeral5}-B, the convergence of the algorithm is analyzed. Subsequently, in Section \uppercase\expandafter{\romannumeral5}-C, the effectiveness of joint optimization design in clutter suppression is analyzed under different numbers of antennas. The superiority is demonstrated by comparing the method without considering mismatches. Section \uppercase\expandafter{\romannumeral5}-D analyzes the impact of various system parameters on the trade-off between communication and sensing. 

\subsection{Simulation Setup}
The simulation settings are summarized as follows unless otherwise specified. In this section, a collocated narrowband MIMO-ISAC BS configuration is considered, where both the transmitter and receiver utilize ULAs with an equal number of antennas. 
\textcolor{black}{In particular, it is supposed that the target of interest is located at ${{\theta }_{0}}={{0}^{\circ }}$. Two types of clutter are taken into consideration, including a point clutter located at an azimuthal angle of \textcolor{black}{${{50}^{\circ }}$ and a densely distributed clutter block spanning from ${{-26}^{\circ }}$ to ${{-30}^{\circ }}$} in the azimuthal direction.} The clutter power is equivalent in each direction while the total transmit power is ${{P}_{t}}= 43$ \ \text {dBm}. The communication channel is modeled in terms of a flat Rayleigh fading.
The number of communication users is assumed to be $K$ = 2, and the similarity coefficient is set to $\alpha =1$. 
The linear frequency modulated (LFM) waveform is a commonly employed reference waveform owing to its favorable characteristics in pulse compression, ambiguity mitigation, and the ability to effectively distinguish point targets \cite{ cui2013mimo, richards2014fundamentals}. \textcolor{black}{ For a time index $n$, the reference waveform can be formulated as
	\begin{equation}
		{\mathbf{x}_{0}}[n]=  { \sqrt{ \frac{{P}_{t}} {M}}} {[1,\ldots, e^{j \pi {( m-1)^2 }/{M}}, \ldots,  e^{j \pi {( M-1)^2 }/{M}} ] }, \forall m.  
	\end{equation} 
	Thus, the covariance matrix of the reference waveform can be obtained by $ \mathbf{R}_0 =\mathbb{E}\left[\mathbf{x}_0 [n]\mathbf{x}^H_0[n] \right] $, where ${\mathbf{x}_{0}}[n] \in \mathbb{C}^{{M}\times {1}}$ represents the $M$ elements of reference waveform at time index $n$. }
The performance of the ISAC system is evaluated by averaging the results of $10^3$ Monte Carlo simulations for every value of $M$ and $K$. The transmit ISAC signal, which includes both sensing waveforms and communication symbols, is generated using random quadrature-phase-shift keying modulation. The block size of the ISAC signal is set to $N$=1024. The MATLAB CVX toolbox is used to solve the problem (31).  \textcolor{black}{Unless specified otherwise, the parameters are summarized in Table IV.}

\begin{table}[t] \color{black}
	\centering
	\setlength{\extrarowheight}{2pt}
	\caption {SIMULATION PARAMETERS}
	\label{table2}
	\begin{tabular}{l|l} 
		\bottomrule  
		\textbf{Parameters}                & \textbf{Value}  \\ \hline	
	The number of antennas for the ISAC BS & $N_t = N_r = M = 16 $  \\ \hline
		The number of communication UEs    & $K = 2$          \\ \hline 
		Location of the target             & $\theta_{0} = {{0}^{\circ }} $ \\ \hline
		Location of the clutter            & $\theta_{p} \in [{50}^{\circ} , {-26}^{\circ} \sim {-30}^{\circ}]  $      \\ \hline
		Total transmit power of ISAC BS    & ${{P}_{t}}= 43$ \ \text {dBm} \\ \hline
		Noise power at user $k$            & $\sigma_{k}^{2} =-80 \ \text {dBm} $  \\ \hline
		Noise power as the receiving BS    & $\sigma^{2} = -80 \ \text {dBm} $ \\ \hline
		Channel power gain                 & $ |\alpha_{0}|^2 = 1$  \\ \hline
		Error tolerance factor             & $ \vartheta  = 10^{-3}$  \\ \hline
	The SINR threshold of communication UE & ${\Gamma } = 5 \ \text{dB}$  \\ \hline 
	The similarity coefficient             & $\alpha =1$            \\ \hline
		The width of beampattern nulls     & $ \Delta$ = 0.03          \\ \toprule 
	\end{tabular}
\end{table}

	\subsection{Algorithm Investigation}
	
	Firstly, we focus on assessing the convergence of the proposed joint algorithm. To characterize the convergence behavior of sensing performance,  the MSE between the obtained beampatterns and the desired beampatterns is used to measure the sensing performance, as illustrated in Fig. 2.
The beampattern resulting from the joint design of transmit beamforming and receive filters can be described as the receive signal power in the direction $\theta \in \left[ -\frac{\pi }{2},\frac{\pi }{2} \right)$, denoted by
\begin{equation}
	P\left( \theta ,\hat{\mathbf{R}} \right) ={{\left| {{\boldsymbol{\omega }}^{H}}\mathbf{A}\left( \theta  \right)\mathbf{x} \right|}^{2}}.	
\end{equation} 
Thus, for $N$ sampled angle grids, the MSE can be defined as follows.
\begin{equation}
	\text{MSE}=\frac{1}{N}\sum\limits_{l=1}^{N}{{{\left| P\left( {{\theta }_{l}};{{\mathbf{R}}_{0}} \right)-P\left( {{\theta }_{l}};\hat{\mathbf{R}} \right) \right|}^{2}}},
\end{equation}
where $\mathbf{R}_0$ represents the covariance matrix of the reference waveform.

\textcolor{black}{ As the number of antennas increases, the receiver benefits from greater spatial diversity provided by multiple receiving antennas, resulting in improved interference suppression and enhanced extraction of target echoes. Moreover, a larger number of antennas provides increased degrees of freedom, enabling the transmission of a more directional and concentrated beam while dispersing interference energy more effectively, thereby improving sensing performance. Consequently, as the number of antennas increases, the MSE converges to a lower value, signifying improved sensing accuracy in the ISAC system. }

	\begin{figure}
		\centerline{\includegraphics [scale=0.6]{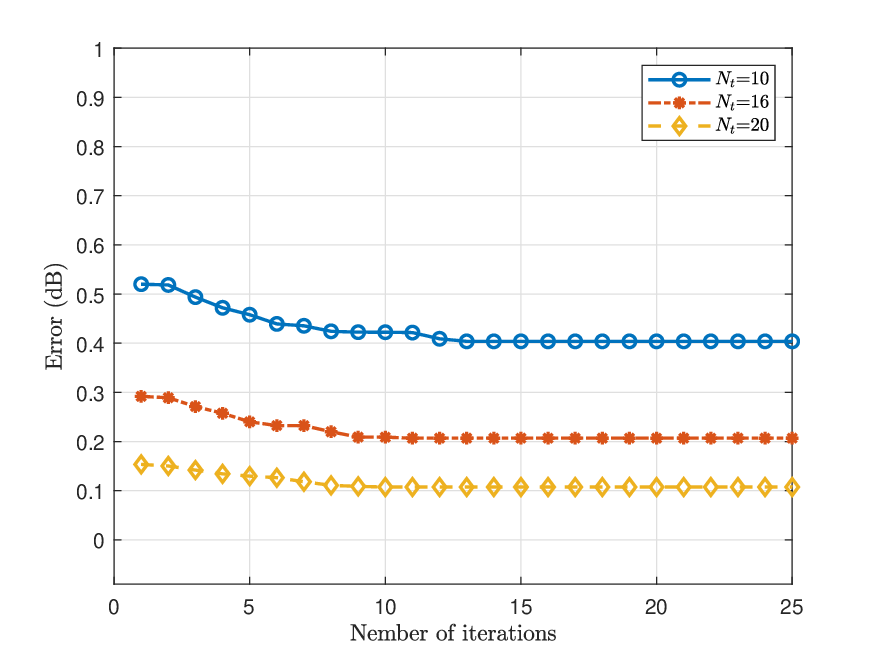}}	
		\caption{\textcolor{black}{Convergence performance of the MSE of beampatterns versus the iteration number for different numbers of the antennas, $N_{t}$ = 10, 16, 20. }}
		\label{Fig. 2}
	\end{figure}

	\subsection{Clutter Suppression}

	The interference suppression performance of the proposed joint optimization algorithm is demonstrated in Fig. 3 and Fig. 4. 
	\textcolor{black}{When a mismatched signal steering vector exists in the clutter direction of the received array response, the clutter covariance matrix can be rewritten by the CMT approach. }In Fig. 3, we compare our proposed transceiver scheme with the constant covariance matrix of clutter as studied in \cite{chen2022generalized} against the optimized robust transceiver design proposed in this paper. The effects of clutter suppression with and without introducing the MZ taper using different numbers of antennas are illustrated in Fig. 4.

\textcolor{black}{ Fig. 3 evaluates two robust transceiver designs in terms of their performance in target detection with $N_t = 16$. By utilizing the reference waveforms $\mathbf{X}_0$, optimal beampatterns from the radar's perspective are provided as a performance benchmark. Although there is a slight reduction in the peak-to-sidelobe ratio, the main beam width remains constant, indicating closely matched optimized beampatterns that validate the effectiveness of our proposed algorithm. 
In contrast, the proposed algorithm demonstrates quicker sidelobe convergence and enhanced clutter nulling, affirming its information transmission capability while maintaining ideal radar beampatterns. Notably, the extended clutter interaction range suggests a higher tolerance for mismatched steering vectors compared to the constant clutter described in [17].}

Fig. 4 (a-c) illustrate the receive beampatterns under mismatched steering vectors for antenna numbers of 10, 16, and 20, respectively. In contrast to the received beampatterns shown in Fig. 4 (d-f) that do not consider mismatched vectors, our proposed joint design approach exhibits notable advantages in effectively suppressing clutter. On one hand, the joint optimization algorithm demonstrates enhanced clutter suppression capability, leading to a reduction of null depth at clutter directions by approximately 10 dB. Particularly, it significantly enhances robustness against mismatches in the steering vector, emphasizing its potential as a dependable and efficient solution in practical clutter suppression scenarios. On the other hand, through a comparison of the receive beampatterns for varying numbers of antennas, it is apparent that the array output gain increases with an increasing number of antennas, resulting in more pronounced clutter suppression. \textcolor{black}{Fig. 4 demonstrates the effectiveness of the CMT technique in mitigating clutter despite mismatched steering vectors in ISAC systems.}
	
	\begin{figure}
		\centerline{\includegraphics [scale=0.65]{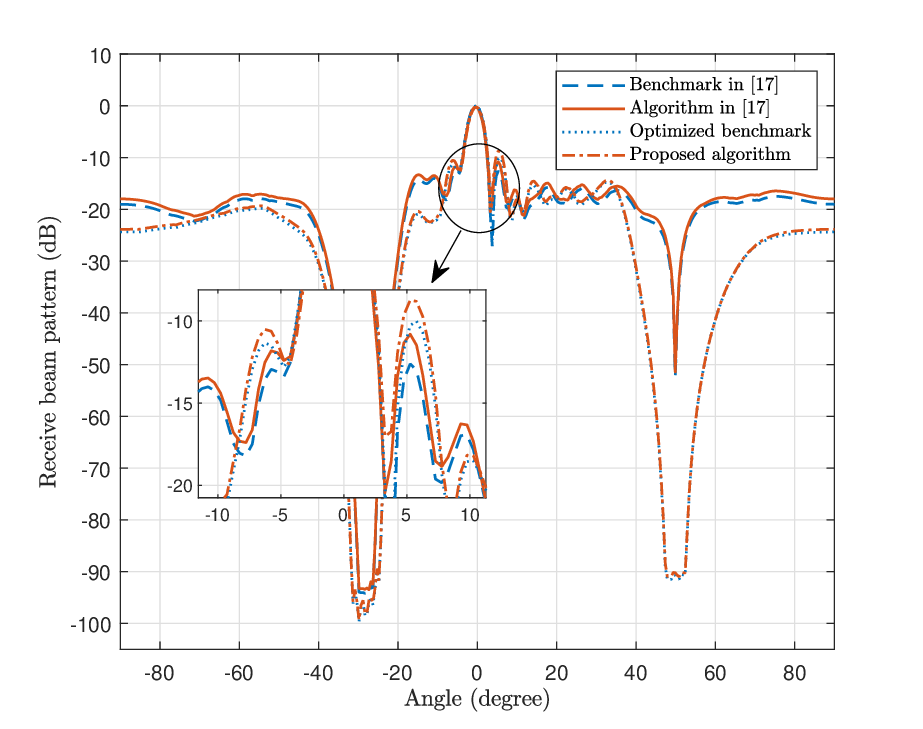}}	
		\caption{ \textcolor{black}{ Comparison of receive beampattern between the transceiver design using constant clutter in [17] and the proposed robust transceiver design incorporating signal-dependent clutter.} }
		\label{Fig. 3}
	\end{figure}
	
	\begin{figure}
		\centerline{\includegraphics [scale=0.55]{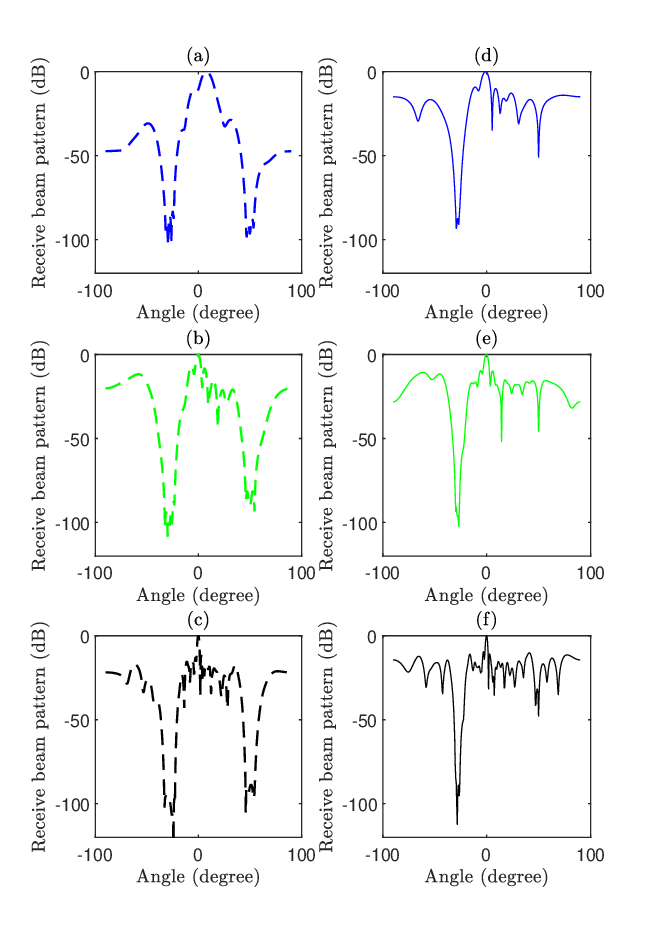}}
		\caption{\textcolor{black}{Comparison of the impact of mismatched steering vectors on the receive beampatterns for different values of $N_t$ is conducted under two conditions: (a-c) Applying the proposed algorithm with MZ taper for $N_t$ = 10, 16, and 20, respectively. (d-f) Employing the method described in [22] without considering mismatched steering vectors for $N_t$ = 10, 16, and 20, respectively. }}  
		\label{Fig.4} 
	\end{figure}

\textcolor{black}{It is essential to highlight that the parameter setting of $\vartriangle$ significantly affects the effectiveness of clutter suppression. Increasing the value of $\vartriangle$ expands the dithering area, which improves error tolerance. However, this expansion results in a reduction in the output SCNR for sensing. The influence of different parameters on communication and sensing will be thoroughly analyzed in subsection V-D. }
	
	\subsection{Parameter Analysis}
	
	The correlation between parameter variations and the interplay between sensing and communication is critical for effective clutter suppression. Fig. 5 demonstrates the trends of the sensing output SCNR and communication sum rate as the parameter $\Delta$ changes where the achievable sum rate is assessed by
	\begin{equation}
		C\left( \gamma  \right)=\sum\limits_{k=1}^{K}{{{\log }_{2}}\left( 1+{{\gamma }_{k}} \right)}.
	\end{equation}

In Fig. 5, as the parameter $\Delta$ increases, the suppression effectiveness against point clutter extends to neighboring angular domains, enhancing robustness against interference at the expense of a partial loss in echo gain. With more transmit antennas, the spatial gain offsets power loss, leading to a slower decrease in SCNR.
Furthermore, Fig. 5 illustrates that adjusting parameter $\Delta$ does not impact the communication sum rate. This observation stems from the fact that environmental noise or clutter does not directly affect communication users. 
\textcolor{black}{ Furthermore, when the number of antennas increases to 20, with $\Delta$ set to 0.06, the echo intensity can still be maintained above 25 dB. However, when the number of antennas decreases to 12, achieving the same level of echo intensity requires $\Delta$ to be no more than 0.035. In other words, increasing spatial degrees of freedom can enhance robustness in clutter suppression.}

	\begin{figure}
		\centerline{\includegraphics [scale=0.6]{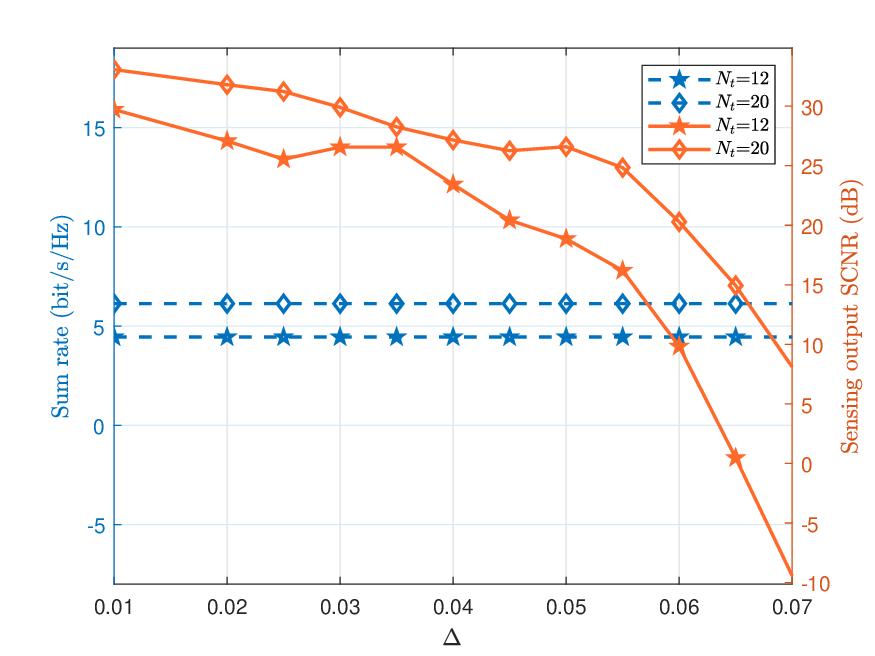}}
		\caption{ \textcolor{black}{ Sensing output SCNR and communication sum rate versus $\Delta $, for $N_t$ = 12, 20. }}
		\label{Fig.5}
	\end{figure}
	
\textcolor{black}{ In Fig. 6, a detailed comparison of the impact of different antenna numbers on sensing output SCNR reveals that before a rapid decline, the sensing SCNR experiences a slight increase, and this increase corresponds to larger values of $\Delta$ with more antennas. This further emphasizes the critical role of increasing the number of antennas for clutter suppression and highlights the importance of investigating the width of clutter nulling. }
The simulation results for different numbers of antennas indicate that within a reasonable range of $\Delta$, clutter suppression can be enhanced with only a minor reduction in SCNR, thereby improving robustness against mismatched steering vectors. 
Simulation results suggest that the optimal range for $\Delta$ typically falls between 0.03 and 0.04. Thus, the proposed joint optimization algorithm achieves robust-enhanced clutter suppression by adjusting $\Delta$. %
	
	\begin{figure}
	\centerline{\includegraphics [scale=0.6]{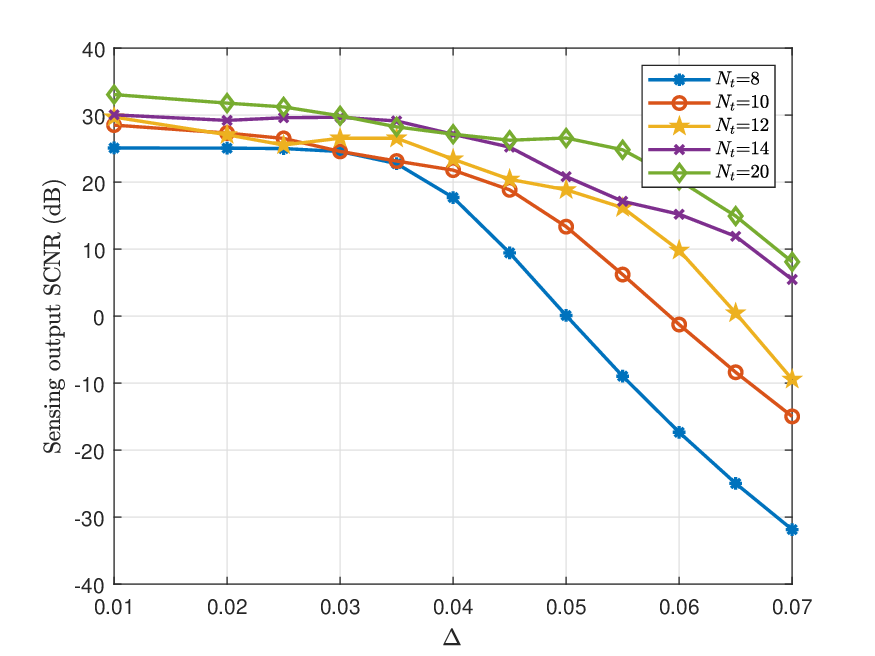}}
	\caption{\textcolor{black}{ Sensing output SCNR versus $\Delta $, for $ N_t$ = 8, 10, 12, 14, 20. }}
		\label{Fig.6}
	\end{figure}
	
	\textcolor{black}{ In optimization problem (21), the parameter $\Gamma$ impacts the minimum received SINR that a single user can achieve, thereby constraining the communication sum rate of the ISAC system. In power-constrained scenarios, an increase in $\Gamma$ results in higher transmission power being used to transmit communication signals, consequently suppressing the sensing output SCNR. Moreover, a larger number of antennas exacerbates this suppression effect. Nonetheless, as shown in Fig. 7, a higher SCNR can still be achieved by increasing the number of transmitting antennas.
	A reduction in the number of antennas leads to decreased spatial diversity, resulting in reduced interference and system complexity. Moreover, the computational complexity of $\gamma_k$ depends on $\mathcal{O}\left( {N}^{2}{{M}^{2}} \right)$, and SCNR depends on $\mathcal{O}\left(2 {N}^{2}{{M}^{2}} \right)$. Therefore, it can be concluded that when the number of transmitting antennas is large, the calculation of SCNR is more affected. }
	
	\begin{figure}
		\centerline{\includegraphics [scale=0.6]{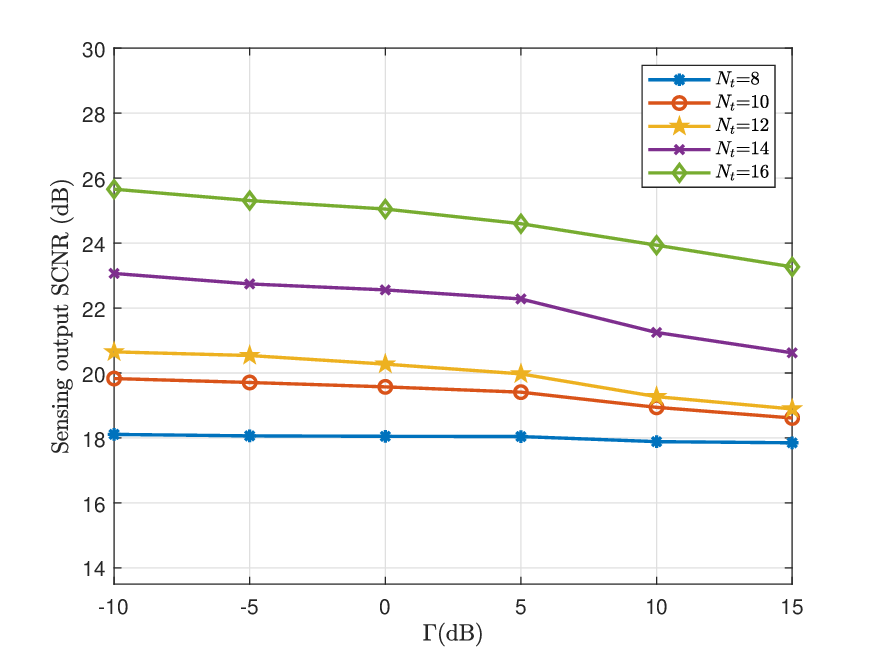}}
		\caption{ \textcolor{black} {Sensing output SCNR versus $\Gamma $, for $ N_t$ = 8, 10, 12, 14, 16. }}
		\label{Fig.7}
	\end{figure}
	
	\begin{figure}
		\centerline{\includegraphics [scale=0.6]{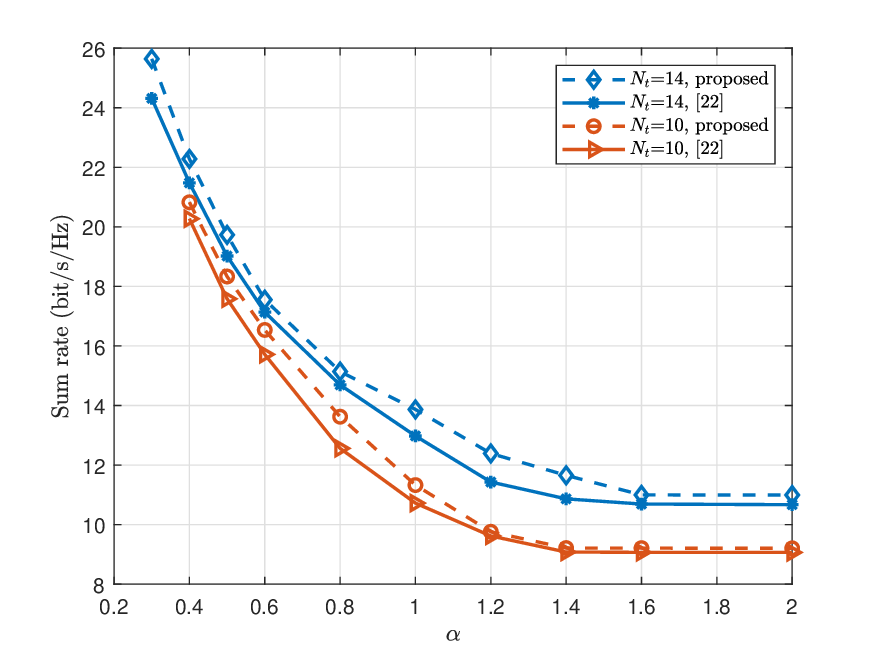}}
		\caption{ \textcolor{black}{Comparison of communication sum rates between the proposed algorithm and the algorithm described in [22] versus parameter $\alpha $, for $N_t$ = 10, 14.}}
		\label{Fig.8}
	\end{figure}

	 \textcolor{black}{ Fig. 8 illustrates the variation of communication sum rate versus waveform similarity parameter $\alpha$ with different number of antennas.
	 As the value of $\alpha$ increases, the optimization problem imposes fewer constraints on sensing, resulting in a decrease in the communication sum rate. As a result, the communication sum rate gradually decreases and stabilizes at a lower value with increasing $\alpha$.
	 Fig. 8 also includes a comparison between the proposed algorithm and the algorithm in [22] that neglects mismatches in the steering vector. In reference [22], the clutter direction exhibits a fixed angular deviation, resulting in a fixed suppression interval for errors caused by clutter uncertainty. However, in our proposed algorithm, the clutter suppression interval is adjustable. Our proposed algorithm exhibits a slight advantage over [22] regarding communication transmission rates when $\Delta$ is set to 0.03. As $\alpha$ approaches 0, the optimized waveform gradually converges to the reference waveform, but this reduces the flexibility in waveform design. Since our optimization objective is to maximize sensing SCNR, lower values of $\alpha$ correspond to higher communication sum rates. As $\alpha$ increases, the communication sum rate decreases and eventually stabilizes.  Overall, the impact of clutter-induced mismatches on communication sum rate gradually decreases and the communication sum rate stabilizes as $\alpha$ increases. By adjusting the value of $\alpha$, the ISAC transceiver can be optimized for specific scenarios.}
	
	\section{Conclusion}
		In this paper, a joint design for transmit beamforming and receive filter of MIMO-ISAC systems is proposed for interference management, including MUI, SI, and clutter.  
	The joint optimal designs are derived to maximize the sensing output SCNR while meeting the communication QoS requirements under constant modulus power constraint and waveform similarity constraint. 
	The problem is decomposed into two subproblems, which are optimized alternately to obtain the optimal solution. Firstly, the receive filter is optimized by reconstructing the interference covariance matrix using a diagonal loading approach. Subsequently,  the transmit beamforming is optimized and decomposed into communication beamforming and sensing beamforming. The two subproblems are alternately optimized to obtain the robustness-enhanced transceiver design together with the combined gain of interference suppression. The performance of the proposed methods is evaluated via numerical results.

	\bibliographystyle{IEEEtran}
	\bibliography{ref}
\end{document}